\documentclass[10pt]{article}
\usepackage[T1]{fontenc}
\usepackage[american]{babel}
\usepackage[margin=1in]{geometry}
\usepackage{titling}
\usepackage{amsmath,amssymb,amsthm,mathtools}
\usepackage{csquotes}
\usepackage{microtype}
\usepackage{tikz}
\usetikzlibrary{arrows.meta,backgrounds,fit,positioning}
\usepackage[hidelinks]{hyperref}
\usepackage[
  backend=biber,
  style=authoryear,
  giveninits=true,
  maxcitenames=2,
  maxbibnames=99,
  uniquename=false,
  doi=false,
  isbn=false,
  url=false,
  eprint=false
]{biblatex}

\title{Unsecured Lending via Delegated Underwriting}
\author{Diego Estevez\\
Divine Research\\
\texttt{diego@divine.inc}}
\date{May 2026}

\hypersetup{
  pdftitle={Unsecured Lending via Delegated Underwriting},
  pdfauthor={Diego Estevez},
  pdfsubject={Pseudonymous credit and delegated underwriting},
  pdfkeywords={pseudonymous lending, delegated underwriting, credit networks, Sybil resistance}
}

\newtheoremstyle{compactplain}
  {4pt plus 1pt minus 1pt}
  {4pt plus 1pt minus 1pt}
  {\itshape}
  {}
  {\bfseries}
  {.}
  {.5em}
  {}

\theoremstyle{compactplain}
\newtheorem{theorem}{Theorem}
\newtheorem{definition}{Definition}

\AtBeginBibliography{\small}
\newcommand{\BaseBudget}[1]{\hat{E}_{#1}}
\newcommand{\Budget}[1]{E_{#1}}
\newcommand{\Earned}[1]{G_{#1}}
\newcommand{\Principal}[1]{x_{#1}}
\newcommand{\LoanAmount}[1]{L_{#1}}
\newcommand{\CreditLimit}[1]{c_{#1}}
\newcommand{\Deleg}[2]{a_{#1\to#2}}
\newcommand{\RequiredDeleg}[1]{R_{#1}}
\newcommand{\Absorb}[1]{\delta_{#1}}
\newcommand{\Loss}{\ell}
\newcommand{\LossOut}{\ell_{\mathrm{out}}}
\newcommand{\Payout}[1]{\operatorname{delegation\,payout}\!\left(#1\right)}

\tikzset{
  delegation-tree/.style={
    font=\small,
    node distance=9mm and 12mm
  },
  delegation-node/.style={
    draw=black,
    line width=0.35pt,
    rounded corners=4pt,
    minimum width=1.7cm,
    minimum height=0.8cm,
    align=center,
    inner sep=2pt
  },
  seed-node/.style={
    delegation-node,
    fill=gray!25,
  },
  user-node/.style={
    delegation-node,
    fill=white,
  },
  delegation-edge/.style={
    draw,
    semithick,
    -{Latex[length=2.6mm,width=1.6mm]},
    shorten <=1pt,
    shorten >=3pt
  },
  forest-box/.style={
    draw=gray!55,
    dashed,
    rounded corners=7pt,
    inner sep=9pt
  }
}

\begin{document}

\maketitle

\begin{abstract}
We develop a mechanism for unsecured lending among pseudonymous users that does not rely on collateral, legal identity, or centralized underwriting. New borrowers enter only through sponsors who delegate part of their own credit capacity, so onboarding a new account reallocates existing borrowing power rather than minting new capacity. Default losses flow back along the sponsor path, while repayment creates earned credit that expands future borrowing capacity. We prove that delegation conserves aggregate credit capacity, that revocation and default remain local to a unique sponsor path, and that a simple cap on earned-credit growth makes repay-then-default weakly unprofitable.
\end{abstract}


\section{Introduction}

Unsecured lending among pseudonymous users poses a fundamental identity problem. When identities are cheap to create and enforcement is weak, lenders cannot rely on collateral, legal recourse, or durable real-world reputation. A borrower who defaults can reappear under a fresh pseudonym and repeat the behavior, steadily draining the protocol.

We address this problem by making borrowing capacity both scarce and earned. New capacity arises only through repayment, so any new borrower must enter through an existing participant who delegates part of their own capacity. The same sponsor path also determines how losses are absorbed: default first extinguishes earned credit and then contracts delegated capacity upstream to a seed.

We assume a set of highly creditworthy seed accounts that provide the initial capacity that bootstraps the network. The mechanism then specifies how this capacity is redistributed across sponsor links, expanded through repayment, and contracted through default.

\noindent\textbf{Contributions.} The paper establishes three main results. First, delegation conserves aggregate credit capacity, so splitting into multiple pseudonyms cannot create new borrowing power. Second, revocation and default are local and well-defined along a unique sponsor path. Third, post-repayment credit growth is capped to make repay-then-default weakly unprofitable.

Section~\ref{sec:related} reviews related work on peer monitoring, credit networks, and Sybil resistance. Section~\ref{sec:mechanism} formalizes the mechanism and proves the main results.

\section{Related Work}\label{sec:related}

Delegated underwriting sits at the intersection of work on peer monitoring, credit networks, Sybil resistance, and unsecured lending.

Peer-monitoring models show that, when collateral is unavailable, local peers can mitigate adverse selection and moral hazard by bearing local default costs \parencite{stiglitz1990}. Our sponsor structure follows that logic: participants decide whom to onboard, and downstream default harms the sponsor path. Unlike standard joint liability, however, exposure is capped by earned credit and delegated capacity rather than open-ended legal liability. This local loss rule is also related to contagion models in which spillovers depend on intermediary capitalization \parencite{elliott2014}.

Credit-network models provide a second connection. Early proposals route money and credit through bilateral trust chains \parencite{fugger2004}, and later work shows that trust graphs can support liquidity without centralized intermediation \parencite{dandekar2011,ramseyer2020}. We retain that graph perspective, but restrict trust to a rooted sponsor forest and let successful repayment create new credit.

The mechanism also draws on graph-based Sybil resistance and on the limits of transitive-trust accounting. A Sybil region can draw only on the limited trust extended to it \parencite{yu2006}, while strong transitive trust is difficult to reconcile with Sybil-proof accounting \parencite{seuken2014}. Our conservation result formalizes the same intuition: extra pseudonyms cannot create additional capacity.

Taken together, these literatures point to delegated underwriting as a natural design for pseudonymous credit: local agents screen, trust links allocate scarce capacity, fake identities face real scarcity, and losses stay on the path that created the exposure.

\section{Mechanism}\label{sec:mechanism}

Let $U$ be the set of users and let $S \subseteq U$ be the set of seeds. A delegation from $u$ to $v$ is a directed edge $(u \to v)$ with amount $\Deleg{u}{v} > 0$; it transfers that amount of credit capacity from $u$ to $v$.

Seeds are treated as highly creditworthy and have no sponsors. Every non-seed $v$ has exactly one sponsor $p(v)$. The delegation graph is therefore a rooted forest, so onboarding and default losses follow a unique path back to a seed. This restriction sacrifices some path diversity relative to general credit networks, but it keeps revocation straightforward and makes loss allocation unambiguous.

\begin{figure}[htbp]
\centering
\begin{tikzpicture}[delegation-tree]
  \node[seed-node] (s1) {Seed $s_1$};
  \node[user-node, below left=of s1] (u1) {$u_1$};
  \node[user-node, below right=of s1] (u2) {$u_2$};
  \node[user-node, below=of u1] (u3) {$u_3$};

  \node[seed-node, right=38mm of s1] (s2) {Seed $s_2$};
  \node[user-node, below=of s2] (u4) {$u_4$};
  \node[user-node, below=of u4] (u5) {$u_5$};

  \draw[delegation-edge] (s1) -- (u1);
  \draw[delegation-edge] (s1) -- (u2);
  \draw[delegation-edge] (u1) -- (u3);
  \draw[delegation-edge] (s2) -- (u4);
  \draw[delegation-edge] (u4) -- (u5);

  \begin{scope}[on background layer]
    \node[forest-box, fit=(s1) (u1) (u2) (u3), label={[font=\footnotesize\bfseries]above:Tree 1}] {};
    \node[forest-box, fit=(s2) (u4) (u5), label={[font=\footnotesize\bfseries]above:Tree 2}] {};
  \end{scope}
\end{tikzpicture}
\caption{Sponsor forest. Shaded nodes are seeds, and arrows denote delegations.\label{fig:delegation-forest}}
\end{figure}
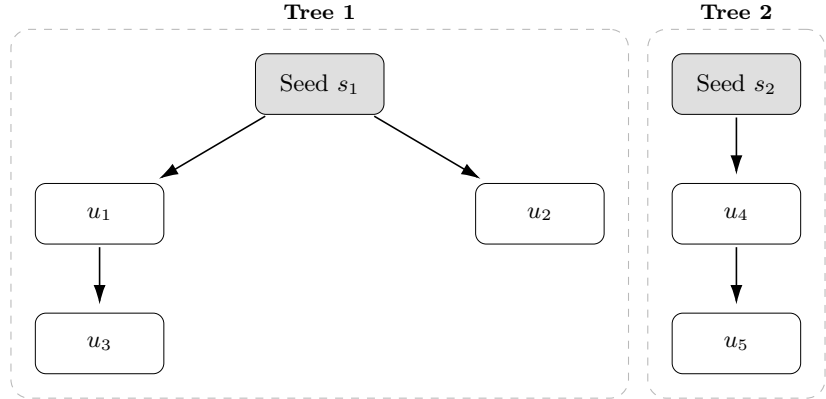

\subsection{Budgets and Credit Limits}

For each user $u \in U$, let $\Budget{u}$ denote current budget, $\Earned{u}$ earned credit from past repayments, and $\Principal{u}$ outstanding principal; $\Principal{u}=0$ when $u$ has no active loan. Each seed $s \in S$ has base budget $\BaseBudget{s} > 0$. Budgets are

\[
\Budget{s} \coloneqq \BaseBudget{s} + \Earned{s}
\quad \text{for } s \in S,
\qquad
\Budget{v} \coloneqq \Deleg{p(v)}{v} + \Earned{v}
\quad \text{for } v \notin S.
\]

When $\Principal{u}=0$, $u$ may borrow up to the credit limit $\CreditLimit{u}$, namely the portion of budget not already delegated:

\[
\CreditLimit{u} \coloneqq \Budget{u} - \sum_{(u\to v)} \Deleg{u}{v}.
\]

After repayment, earned credit updates according to

\[
\Earned{u} \leftarrow \Earned{u} + \Delta G_u.
\]

\begin{theorem}[Credit conservation under delegation]\label{thm:credit-conservation}
Delegation conserves aggregate credit capacity; it only redistributes existing capacity across accounts.
\end{theorem}
\begin{proof}
Summing $\CreditLimit{u} = \Budget{u} - \sum_{(u\to v)} \Deleg{u}{v}$ over all users gives
\[
\sum_{u\in U} \CreditLimit{u} = \sum_{u\in U} \Budget{u} - \sum_{(u\to v)} \Deleg{u}{v}.
\]
Also,
\[
\sum_{u\in U} \Budget{u} = \sum_{s\in S} \BaseBudget{s} + \sum_{u\in U} \Earned{u} + \sum_{v\notin S} \Deleg{p(v)}{v}.
\]
Every delegation appears exactly once in the last sum, namely as the incoming delegation of its non-seed child, so
\[
\sum_{v\notin S} \Deleg{p(v)}{v} = \sum_{(u\to v)} \Deleg{u}{v}.
\]
Substituting cancels the delegation terms and yields
\[
\sum_{u\in U} \CreditLimit{u} = \sum_{s\in S} \BaseBudget{s} + \sum_{u\in U} \Earned{u}.
\]
Hence aggregate credit capacity depends only on seed base budgets and earned credit. Delegation only redistributes that capacity, so additional pseudonyms cannot increase it.
\end{proof}

\subsubsection{Revocation}

Revoking the sponsor edge into $v$ is admissible only if the subtree rooted at $v$ remains solvent.

\begin{definition}[Required delegation]\label{def:required-delegation}
For each non-seed $v$, the required delegation is given recursively by
\[
\RequiredDeleg{v} \coloneqq \max\left\{0,\; \Principal{v} + \sum_{(v\to w)} \RequiredDeleg{w} - \Earned{v}\right\}.
\]
It is the minimum delegation that must remain on the sponsor edge into $v$ to keep $v$ and all of its descendants solvent.
\end{definition}

\begin{theorem}[Revocation solvency]\label{thm:revocation-solvency}
Let $u=p(v)$. If $\Deleg{u}{v}'$ is the delegation remaining on edge $(u\to v)$ after revocation, then the subtree rooted at $v$ remains solvent if and only if
\[
\Deleg{u}{v}' \ge \RequiredDeleg{v}.
\]
\end{theorem}
\begin{proof}
After revocation, the subtree rooted at $v$ is solvent exactly when $v$ has enough remaining budget to cover its own principal and the minimum required delegations to all children. Since $v$'s remaining budget is $\Deleg{u}{v}' + \Earned{v}$, this condition is
\[
\Deleg{u}{v}' + \Earned{v} \ge \Principal{v} + \sum_{(v\to w)} \RequiredDeleg{w}.
\]
Rearranging gives
\[
\Deleg{u}{v}' \ge \Principal{v} + \sum_{(v\to w)} \RequiredDeleg{w} - \Earned{v}.
\]
Since delegation is nonnegative, this is equivalent to
\[
\Deleg{u}{v}' \ge \max\left\{0,\; \Principal{v} + \sum_{(v\to w)} \RequiredDeleg{w} - \Earned{v}\right\}
= \RequiredDeleg{v},
\]
where the equality is Definition~\ref{def:required-delegation}.
\end{proof}

\subsection{Default propagation}

When borrower $u$ defaults on principal $\Principal{u} > 0$, the protocol first burns the borrower's earned credit and records the residual loss:

\[
\Absorb{u} \coloneqq \min\left\{\Earned{u}, \Principal{u}\right\},
\qquad \Earned{u} \leftarrow \Earned{u} - \Absorb{u},
\qquad \Loss \coloneqq \Principal{u} - \Absorb{u}.
\]

The residual loss then propagates up the sponsor path. Starting from $j \coloneqq u$, while $\Loss > 0$ and $j$ is not a seed, let $v \coloneqq p(j)$ and apply

\[
\begin{aligned}
\Deleg{v}{j} &\leftarrow \Deleg{v}{j} - \Loss,\\
\Absorb{v} &\coloneqq \min\left\{\Earned{v}, \Loss\right\},\\
\Earned{v} &\leftarrow \Earned{v} - \Absorb{v},\\
\Loss &\leftarrow \Loss - \Absorb{v},\\
j &\leftarrow v.
\end{aligned}
\]

If the loop ends at a seed $s$ with $\Loss > 0$, the protocol charges the remainder to the seed's base budget:

\[
\BaseBudget{s} \leftarrow \BaseBudget{s} - \Loss.
\]

Finally set $\Principal{u} \leftarrow 0$.

\begin{theorem}[Default propagation is well-defined]\label{thm:default-well-defined}
Every delegation reduction along the default path is feasible, and any residual loss reaching a seed is absorbed without overdrawing base budget.
\end{theorem}
\begin{proof}
We claim that whenever the loop is at a non-seed child $j$ with current residual loss $\Loss$, one has
\[
\Loss \le \RequiredDeleg{j}.
\]
For $j=u$, after burning the borrower's earned credit,
\[
\Loss = \Principal{u} - \min\left\{\Earned{u}, \Principal{u}\right\}
= \max\left\{0,\; \Principal{u} - \Earned{u}\right\}
\le \RequiredDeleg{u},
\]
because $\sum_{(u\to w)} \RequiredDeleg{w} \ge 0$ in Definition~\ref{def:required-delegation}.

Now suppose the loop is at $j$ with $\Loss \le \RequiredDeleg{j}$, and let $v=p(j)$. After updating $v$, the outgoing residual loss is
\[
\LossOut = \Loss - \min\left\{\Earned{v}, \Loss\right\}
= \max\left\{0,\; \Loss - \Earned{v}\right\}.
\]
Since $\Principal{v} + \sum_{(v\to w)} \RequiredDeleg{w} \ge \RequiredDeleg{j}$, we get
\[
\LossOut
\le \max\left\{0,\; \RequiredDeleg{j} - \Earned{v}\right\}
\le \max\left\{0,\; \Principal{v} + \sum_{(v\to w)} \RequiredDeleg{w} - \Earned{v}\right\}
= \RequiredDeleg{v}.
\]
So the same bound holds one step higher on the sponsor path.

Therefore every delegation reduction is feasible: by Theorem~\ref{thm:revocation-solvency}, the current state satisfies $\Deleg{p(j)}{j} \ge \RequiredDeleg{j}$, hence the loop never subtracts more than the delegation available on edge $(p(j) \to j)$.

If residual loss reaches a seed $s$ through child $j$, the same bound gives
\[
\Loss \le \RequiredDeleg{j} \le \Deleg{s}{j} \le \sum_{(s\to w)} \Deleg{s}{w} \le \Budget{s} = \BaseBudget{s} + \Earned{s}.
\]
After burning seed earned credit, the remaining charge is
\[
\Loss - \min\left\{\Earned{s}, \Loss\right\} = \max\left\{0,\; \Loss - \Earned{s}\right\} \le \BaseBudget{s},
\]
so the seed's base budget is not overdrawn.
\end{proof}

\begin{theorem}[Sponsor-path credit limit conservation]\label{thm:default-credit-conservation}
On default, every upstream node on the seed-to-$u$ path other than the borrower retains its credit limit, and aggregate credit falls by exactly $\Principal{u}$.
\end{theorem}
\begin{proof}
Take any node $v \neq u$ on the sponsor path from $u$ to its seed, and let $j$ be the child of $v$ on that path. When residual loss $\Loss$ reaches $v$, the algorithm reduces $\Deleg{v}{j}$ by $\Loss$, burns
\[
\Absorb{v} = \min\left\{\Earned{v}, \Loss\right\},
\qquad
\LossOut = \Loss - \Absorb{v}
\]
and passes $\LossOut$ upward. Thus $v$'s budget eventually falls by exactly $\Loss$: for a non-seed, earned credit falls by $\Absorb{v}$ and the incoming delegation from $p(v)$ later falls by $\LossOut$; for a seed, earned credit falls by $\Absorb{v}$ and base budget falls by $\LossOut$. In either case,
\[
\Delta \Budget{v} = -(\Absorb{v} + \LossOut) = -\Loss.
\]
Its total outgoing delegation also falls by $\Loss$, since only the edge $(v\to j)$ is changed. Therefore
\[
\CreditLimit{v}' = (\Budget{v} - \Loss) - \left(\sum_{(v\to w)} \Deleg{v}{w} - \Loss\right) = \CreditLimit{v}.
\]
So every upstream node retains its credit limit.

For the borrower $u$, the initial earned-credit burn removes $\Absorb{u}$ from budget and the contraction of the sponsor edge into $u$ removes the remaining $\Principal{u} - \Absorb{u}$. Thus $u$'s budget falls by $\Principal{u}$ in total, while its outgoing delegations are unchanged. Hence
\[
\CreditLimit{u}' = (\Budget{u} - \Principal{u}) - \sum_{(u\to w)} \Deleg{u}{w} = \CreditLimit{u} - \Principal{u}.
\]

All other users are unchanged. Summing over users gives
\[
\sum_{w\in U} \CreditLimit{w}' = \sum_{w\in U} \CreditLimit{w} - \Principal{u}.
\]
Thus a default destroys exactly the unpaid principal in aggregate credit capacity.
\end{proof}

\subsection{Interest}

A proposed loan to borrower $v$ with principal $\LoanAmount{v} > 0$ and term $T$ pays total interest $I \coloneqq I^R + I^D$, where $I^R$ compensates the protocol for expected default loss and $I^D$ compensates sponsors for locked delegation. All state variables below are measured at a fixed pre-origination snapshot, so same-batch revocations, repayments, and redelegations do not affect them.

\subsubsection{Protocol premium}

The protocol premium prices expected default loss. Let $D_v \in (0,1)$ denote the estimated default probability over term $T$, and let $r^R$ denote the protocol premium rate. Then

\[
I^R \coloneqq r^R \LoanAmount{v} T.
\]

\begin{theorem}[Break-even protocol premium rate]\label{thm:protocol-break-even}
The protocol weakly breaks even in expectation if and only if
\[
r^R \ge \frac{D_v}{(1-D_v)T}.
\]
Equality gives exact break-even.
\end{theorem}
\begin{proof}
If the borrower repays, the protocol receives $I^R$; if the borrower defaults, it loses $\LoanAmount{v}$. Hence expected protocol cash flow is $(1-D_v)I^R - D_v \LoanAmount{v} = \LoanAmount{v}\bigl((1-D_v)r^R T - D_v\bigr)$. Since $\LoanAmount{v} > 0$, nonnegativity is equivalent to $(1-D_v)r^R T \ge D_v$, i.e. $r^R \ge D_v/((1-D_v)T)$.
\end{proof}

\begin{theorem}[Repay-then-default bound]\label{thm:repay-then-default}
After repayment, any repay-then-default strategy yields incremental profit of at most $\Delta G_v - I^R$. In particular, if $\Delta G_v \le I^R$, the strategy is weakly unprofitable.
\end{theorem}
\begin{proof}
Repayment has two incremental effects: it costs the borrower $I^R$ immediately and awards earned credit $\Delta G_v$. By definition, that earned credit raises aggregate credit capacity by exactly $\Delta G_v$, so even in the most borrower-favorable case it can create at most $\Delta G_v$ of later defaultable notional. Hence incremental profit from any repay-then-default strategy is at most $\Delta G_v - I^R$, which is nonpositive when $\Delta G_v \le I^R$.
\end{proof}

\subsubsection{Delegation premium}

The delegation premium compensates sponsors for delegation locked along the unique seed-to-borrower path.

\begin{definition}[Delegation utilization]\label{def:delegation-utilization}
Seed-level delegation utilization is
\[
U^D \coloneqq \frac{\sum_{s\in S}\sum_{(s\to v)} \Deleg{s}{v}}{\sum_{s\in S} \Budget{s}} \in [0,1].
\]
\end{definition}

\paragraph{Delegation premium rate.}

Let $\bar U^D$ denote the pre-origination snapshot value of the utilization in Definition~\ref{def:delegation-utilization}. Given a maximum rate $r^D_{\max} > 0$, interpreted as the delegation premium rate at zero utilization, set the delegation premium rate to

\[
r^D \coloneqq r^D(\bar U^D) = r^D_{\max}(1-\bar U^D).
\]

This linear schedule pays the highest delegation premium when seed utilization is low and lowers the rate as utilization rises. The remaining task is therefore to identify, on each edge of the sponsor path, how much delegation the proposed loan actually locks.

\paragraph{Locked delegation.}

Below any sponsor edge, part of the proposed principal $\LoanAmount{v}$ may already be absorbable by earned credit inside the downstream subtree. Only the remainder must be backed by delegation above.

\begin{definition}[Local buffer]\label{def:local-buffer}
Along the sponsor path $u_0 \to u_1 \to \cdots \to u_d = v$, the local buffer at node $u$ is
\[
b_u \coloneqq \max\left\{0,\; \Earned{u} - \Principal{u} - \sum_{(u\to w)} \RequiredDeleg{w}\right\}.
\]
It is the amount of additional downstream principal that $u$ can absorb without increasing the minimum delegation required from its sponsor.
\end{definition}

\begin{definition}[Downstream buffer]\label{def:downstream-buffer}
For each edge $(u_k \to u_{k+1})$ on that path, the downstream buffer below the edge is
\[
B_k \coloneqq \sum_{i=k+1}^{d} b_{u_i}.
\]
\end{definition}

\begin{theorem}[Feasibility and locked delegation from buffers]\label{thm:locked-delegation}
A loan of principal $\LoanAmount{v}$ is feasible along the sponsor path if and only if
\[
\max\left\{0,\; \LoanAmount{v} - B_k\right\} \le \Deleg{u_k}{u_{k+1}}
\qquad \text{for } k=0,\ldots,d-1.
\]
When this holds, the delegation locked on edge $(u_k \to u_{k+1})$ is
\[
m_k \coloneqq \max\left\{0,\; \LoanAmount{v} - B_k\right\}.
\]
\end{theorem}
\begin{proof}
Fix $k$. By Definitions~\ref{def:local-buffer} and~\ref{def:downstream-buffer}, the nodes below edge $(u_k \to u_{k+1})$, namely $u_{k+1},\ldots,u_d$, can jointly absorb at most $B_k$ of the proposed loan. So the amount that must still be backed from above that edge is
\[
\max\left\{0,\; \LoanAmount{v} - B_k\right\}.
\]
The loan is feasible exactly when this required amount does not exceed the delegation available on the edge. Under that condition, the amount rendered unavailable to other uses on the edge is precisely that required amount, namely $m_k$.
\end{proof}

Thus pricing reduces to Definitions~\ref{def:local-buffer} and~\ref{def:downstream-buffer}: compute each node's local buffer $b_u$, sum them below each edge to obtain $B_k$, and charge sponsors on the leftover amount $m_k$ that downstream buffers do not absorb.

Each sponsor receives

\[
\Payout{u_k} \coloneqq r^D m_k T.
\]

The borrower pays in total
\[
I^D \coloneqq \sum_{k=0}^{d-1} \Payout{u_k} = r^D T \sum_{k=0}^{d-1} m_k.
\]

So the delegation premium is budget balanced by construction: every dollar paid by the borrower is paid out to sponsors along the path. Intermediaries increase total premiums only when downstream buffers below their edge are small enough that their delegation is genuinely locked.

\section{Discussion and Scope}

The mechanism should be read as a baseline accounting design rather than a complete lending stack. Its strongest assumption is an exogenous seed layer with positive base budgets. The contribution of delegated underwriting is not to manufacture trust from scratch, but to make access to sponsor capacity scarce, traceable, and loss-bearing.

The analysis is also intentionally static. Loans are summarized by outstanding principal, pricing uses an exogenous default estimate, and the theorems abstract from recovery, servicing costs, correlated shocks, dynamic reputation, and strategic coalition formation among sponsors. Those considerations matter for deployment and may affect welfare or calibration, but they do not alter the core conservation, revocation, and default-propagation identities established here.

\section{Conclusion}

Delegated underwriting does not solve pseudonymous unsecured lending in full generality, but it does isolate a coherent accounting baseline. Borrower entry is open, yet every new exposure must be backed by scarce sponsor capacity; aggregate credit expands only through repayment; and default losses remain on the path that created the exposure. Correlated shocks, dynamic reputation, richer underwriting signals, and collusive sponsor cartels remain open extensions. Even so, the mechanism resolves the central accounting problem: credit capacity can grow through repayment, but not through the creation of fresh pseudonyms.

\printbibliography

\end{document}